\documentclass{elsart}
\usepackage{stmaryrd,euscript,float,bbold,amssymb,amsmath}

\newcommand{\ca}{{\mathcal A}}

\newcommand{\cp}{{\mathcal P}}

\newcommand{\cs}{{\mathcal S}}

\def\part#1{\operatorname{Part}(#1)}

\newtheorem{definition}{Definition}
\newtheorem{lemma}{Lemma}
\newtheorem{corollary}{Corollary}
\newtheorem{theorem}{Theorem}

\newenvironment{proof}{\begin{trivlist}
 \item[\hspace{\labelsep}{\em\noindent Proof: }]
}{\hfill$\Box$\end{trivlist}}
\usepackage{tikz}
\usepackage{hyperref,bookmark}
\newcommand{\intv}[1]{\left \llbracket #1 \right \rrbracket}
\newcommand{\probl}[3]{
\begin{flushleft}
\fbox{
\begin{minipage}{13cm}
\noindent {\sc #1}\\
          {\bf Input:} #2\\
          {\bf Output:} #3
\end{minipage}}
\medskip
\end{flushleft}
}
\begin{document}
\begin{frontmatter}
  \title{Temporal Matching\thanksref{paper}}
\thanks[paper]{Part of the results reported in this paper were presented at \textit{CTW'18}.
  Links to the source code and the GUI of the link stream generator:\\\url{https://github.com/antoinedimitriroux/Temporal-Matching-in-Link-Streams}\\
  \url{https://antoinedimitriroux.github.io}\\
  For financial support, we are grateful to:
  \textit{Thales Communications \& Security}, project TCS.DJ.2015-432;
  \textit{Agence Nationale de la Recherche Technique}, project 2016.0097;
  \textit{Centre National de la Recherche Scientifique}, project INS2I.GraphGPU.}

\author[lip6]{Julien Baste}
\author[lip6]{Binh-Minh Bui-Xuan}
\author[lip6,thales]{Antoine Roux}
\address[lip6]{Laboratoire d'Informatique de Paris 6 (LIP6),
  Centre National de la Recherche Scientifique (CNRS),
  Sorbonne Universit\'e (SU UPMC).}
\address[thales]{Thales Communications \& Security, Thales Group.\\
  \texttt{[julien.baste,buixuan,antoine.roux]@lip6.fr}}

\begin{abstract}
  A link stream is a sequence of pairs of the form $(t,\{u,v\})$, where $t\in\mathbb N$ represents a time instant and $u\neq v$.
  Given an integer $\gamma$, the $\gamma$-edge between vertices $u$ and $v$, starting at time $t$, is the set of temporally consecutive edges defined by $\{(t',\{u,v\}) \mid t' \in \intv{t,t+\gamma-1}\}$.
  We introduce the notion of temporal matching of a link stream to be an independent $\gamma$-edge set belonging to the link stream.
  We show that the problem of computing a temporal matching of maximum size is NP-hard as soon as $\gamma>1$.
  We depict a kernelization algorithm parameterized by the solution size for the problem.
  As a byproduct we also give a $2$-approximation algorithm.

  Both our $2$-approximation and kernelization algorithms are implemented and confronted to link streams collected from real world graph data.
  We observe that finding temporal matchings is a sensitive question when mining our data from such a perspective as: managing peer-working when any pair of peers $X$ and $Y$ are to collaborate over a period of one month, at an average rate of at least two email exchanges every week.
  We furthermore design a link stream generating process by mimicking the behaviour of a random moving group of particles under natural simulation, and confront our algorithms to these generated instances of link streams.
  All the implementations are open source.
\end{abstract}

\end{frontmatter}
\section{Introduction}\label{sec_intro}
The problem of finding a maximum matching is a fundamental and well studied question.
It consists in finding a maximum independent edge set of a given graph.
It can be solved in polynomial time by the well known Edmonds algorithm~\cite{E65}.
On the theoretical side, Edmonds result plays a primary role in combinatorial optimization.
This is not only because it made a major historical impact in pointing out the polytope structure of a graph problem,
but also because this result has marked the beginning of a long and fruitful list of matching algorithms all having polynomial worst case time complexity.
We can cite in this sense the connection between \textsc{Matching} and matrix multiplication which was exploited for algorithm design~\cite{CGS15,MS04}.

A lot of research effort has been put in investigating variants of \textsc{Matching} as well.
For instance, a tricky structural analysis helps in devising a linear time procedure for finding popular matchings~\cite{AIKM07}.
Albeit it must be under some conditions called fairness, the fact that a linear time algorithm exists for this kind of popular matching helps in better understanding the underlying discrete structure of matchings.
Surprisingly, while being a well-known and fundamentally polynomial algorithmic problem, \textsc{Matching} has lately attracted research interest in the parameterized areas of algorithmic as well~\cite{FLSPW18,MNN17}.
Here, the overall effort has been put in reducing the polynomial time complexity to linear time, by means of factorising bits of the time complexity to depend on another parameter of the input instance rather than its size.

On the practical side, \textsc{Matching} is a convenient formalism to approach task management problems.
For instance, in a bipartite graph where one vertex set represents chores and the other vertex set represents executors, each having the ability to execute a (different) subset of chores, \textsc{Matching} models the question of maximising the number of chores that can be executed.
This problem has been intensively investigated under the setting of streaming inputs, where unpredictable arrivals of executors must be affected to chores in real time, see e.g.~\cite{M14,WW15}.
In this topic, a careful randomized study~\cite{DKR16} not only provides a streaming algorithm achieving competitive approximation ratios, but it also gives the upper bound of extra information the streaming algorithm requires, and, particularly, a tricky proof of a lower bound of extra information one need to use in order to achieve the previously said approximation ratio.
More generally, in an arbitrary graph representing compatible coworkers, \textsc{Matching} models the concern of maximising the overall workload when work must be done by compatible pairs.
This problem has recently been investigated from a heuristics point of view~\cite{DKPU18}, as well as under a quantitative comparison of the greedy approach on large input~\cite{WL17}.

From the perspective of mining data collected from human activities, however, the input graph should be taken under the light of the time dimension: graph edges are time stamped edges.
They come ordered by the time instants where they are recorded.
We call this kind of data a link stream, in the sense of~\cite{LVM17,VLM16}.
The most natural illustration of such thing is web logs, where any single line of log includes a field under time format.
Phone calls between individuals are also time dependent information, so are email exchanges.
Other kind of time dependent interaction could arise from peer programming management in IT best practices too.
For instance, let us consider a human resource platform where collaborators register for the coming trimester the time intervals where they are unavailable for work.
For simplicity we discretize these time intervals by days, from $1$ to $90$.
Besides, the collaborators also communicate via keywords the hard skills in which they are efficient.
Let us consider that a task must be processed by two skilled collaborators over at least $\gamma=5$ five consecutive days before delivery.
By human limitations, a collaborator will only process a given maximum number of tasks at a time.
Under these conditions, given a potentially infinite number of tasks to process, how many deliveries can be made for the coming trimester?
How continuously can delivery be?
How dense can peer working be versus individualistic task processing?
Fundamentally, \textit{how to quantify the concession in term of continuous delivery in favoring peer programming over individualistic behaviour?}
While not fully answering to these questions, we propose to make one step toward this kind of reflection by formalizing the notion of timed collaboration, and show how to compute it.

A link stream $L$ is a triple $L=(T,V,E)$ where $E$ is a sequence of pairs of the form $(t,\{u,v\})$, with $\{u,v\}\in {V\choose 2}$ being an edge in the sense of classical loopless undirected simple graphs, and $t\in T\subseteq\mathbb N$ an integer representing a discretized time instant.
If every pair $(t,\{u,v\})$ in $L$ satisfies $t=t_0$ for some fixed $t_0$, then we say that link stream $L$ is a graph.
Given an integer $\gamma$, a time instant $t$, and two distinct vertices $u$ and $v$, we define the $\gamma$-edge between $u$ and $v$ starting at time $t$ as the set $\{(t',\{u,v\}) \mid t' \in \intv{t,t+\gamma-1}\}$.
A temporal vertex is a pair $(t,u)$, representing vertex $u\in V$ at time $t\in T$.
We say that a $\gamma$-edge $\Gamma$ contains a temporal vertex $(t,u)$ if there exists a vertex $v \in V$ such that $(t,\{u,v\}) \in \Gamma$.
Two $\gamma$-edges are independent if there is no temporal vertex that is contained in both of them.
Finally, a $\gamma$-matching of link stream $L$ is a set of pairwise independent $\gamma$-edges where each $\gamma$-edge contains exclusively edges from $L$.
We define \textsc{$\gamma$-matching} as the problem of computing a maximum $\gamma$-matching from a given input link stream.
We believe that problems involving $\gamma$-edges with $\gamma=1$ somewhat are rooted in classical graph theory, whereas $\gamma$-edges for $\gamma>1$ intrinsically model temporal interactions.
For instance, when $\gamma=1$, this problem can be solved by a slight extension of previously mentioned Edmonds algorithm~\cite{E65}.
In recent studies under a temporal perspective, the problem when $\gamma=1$ has also been considered along with additional conditions in the computed $\gamma$-matching, making it NP-hard~\cite{BELP18,MS16}.
In this paper, we focus on $\gamma$-edges with a non-trivial duration, that is, when $\gamma>1$.

Unfortunately, \textsc{$\gamma$-matching} turns out to be NP-hard for $\gamma>1$.
We subsequently address the question of pre-processing, in polynomial time, an input instance of \textsc{$\gamma$-matching}, in order to reduce it to an equivalent instance of smaller size, in the sense of kernelization algorithms introduced in~\cite{DF99}.
We show that \textsc{$\gamma$-matching} when parameterized by the solution size admits a quadratic kernel.
On the way to do so, we also point out a simple way to produce a $2$-approximation algorithm for \textsc{$\gamma$-matching}.

We try to comprehend our result from a practical point of view.
From this perspective we design a link stream generating process\footnote{Direct link to the GUI of the generator:\\\url{https://antoinedimitriroux.github.io}} by mimicking the behaviour of a random moving group of particles, using natural simulation: velocity, friction, and random walk.
The generating process helps us in unit testing our implementations on small generated inputs,
as well as in stress testing our implementations on large inputs mimicking natural movements.

Both our $2$-approximation and kernelization algorithms are implemented\footnote{The source code is available at\\\url{https://github.com/antoinedimitriroux/Temporal-Matching-in-Link-Streams}} and confronted, not only to our generated link streams, but also to two particular sets of link streams collected from real world graph data.
These raw datasets are first cleaned by a procedure that we call time-compression, and argue the need for it right below.
In one dataset the link streams have been built by time-compression over exchanges collected from the Enron emailing network~\cite{KY04}.
The other dataset is built by time-compression over a recording of $2\times80$ minute Rollerblade touring in Paris~\cite{TLBCDW09}.

The reason for us to time-compress the raw data is because, therein, the (machine recorded) consecutive time stamps can happen quite instantaneously for human standards.
This leaves no chance for a $\gamma$-edge to exist, as soon as $\gamma>1$.
For instance with the Enron emails, by ISO8601 time stamps are discretized down to the order of seconds.
However, there is absolutely no chance for that individuals $X$ and $Y$ exchange two different emails in any two consecutive seconds in the whole duration of the experiment:
there is simply no time to read the first email and type a reply the following second.
For Enron we usually merge up the time stamps to the order of half a week:
if an email is received, read, thought over, eventually replied within $3.5$ days, then we consider there is collaboration within that period of time.
From this perspective, we compress our raw data by edge contraction over the time dimension, formally as follows.
For any link stream $L=(T,V,E)$ and $1<\delta<|T|$, we define the \textit{$\delta$-compression} $L_\delta=(T_\delta,V_\delta,E_\delta)$ as
$V_\delta=V$, $T_\delta=\llbracket \frac{\min T}{\delta}, \frac{\max T}{\delta} \rrbracket$, and
$$E_\delta=\{(t,\{u,v\})~|~\exists t'\in T: \delta t\leq t'<\delta(t+1) \land (t',\{u,v\})\in E\}.$$

After running our implementation on the two datasets, we make three observations.
First, we believe that solving \textsc{$\gamma$-matching} is not easy as soon as we time-compress the raw data with ``human-understandable'' values of $\delta$ and $\gamma$.
Second, on small values of $\gamma$,
we observe that the kernelization algorithm helps in reducing the input link stream to an equivalent instance of size approximatively $10-20\%$ the size of the original input.
This gives measurable evidence of performance for our preprocessing by kernelization.
Our third observation is very marginal.
Note beforehand from definition that the $2$-approximation algorithm produces a lower bound for \textsc{$\gamma$-matching}, which is at least half the optimal value.
Moreover, note also that the kernelization algorithm gives a naive upper bound for \textsc{$\gamma$-matching} by simply counting the number of $\gamma$-edges present in the kernel.
Our third observation from the numerical analysis is that these upper bound and lower bound nearly meet on some areas in the datasets.
Even though they remain extremely marginal, these cases point out that, sometimes, a kernelization algorithm can also provide a numerical proof of optimality of the result found by a (greedy) $2$-approximation.
This hints at the usefulness beyond theoretical considerations of \textsc{$\gamma$-matching} kernelization.
Moreover, our kernelization runtime is under ten seconds for inputs where the input size is some hundreds thousand and the parameter is some thousands.
Fixed parameter tractable (FPT) paradigm in general, and kernelization in particular, would never be numerically helpful if the complexity analysis hides a big function of the parameter in the Landau notation.
Luckily, we use simple algorithmic processes for our kernelization.

The paper is organised as follows.
We first introduce the notion of temporal matching in Section~\ref{sec:tm}.
In Section~\ref{sec:aa}, we present our algorithmic tools in order to obtain
our main result, the kernelization algorithm: it is presented in Section~\ref{sec:ka}.
In Section~\ref{sec:num}, we present our numerical analysis.
We close the paper with concluding remarks and directions for further research.

\section{Temporal matching}\label{sec:tm}
Unless otherwise stated, graphs in this paper are simple, undirected and loopless graphs.
We denote by $\mathbb{N}$ the set of non negative integer.
Given two integers $p$ and $q$, we denote by $\intv{p,q}$ the set $\{r \in \mathbb{N} \mid p \leq r \leq q\}$.
A \emph{link stream} $L$ is a triple $(T,V,E)$ such that
$T \subseteq \mathbb{N}$ is an interval,
$V$ is a set, and
$E \subseteq T \times {V \choose 2}$.
The link stream can be seen as an extension of graphs.
Indeed, a graph is a link stream where $|T| = 1$.
The elements of $V$ are called \emph{vertices} and the elements of $E$ are called \emph{(timed) edges}.
A \emph{temporal vertex} of $L$ is a pair $(t,u)$ such that $t \in T$ and $u \in V$.

Given an integer $\gamma$, a $\gamma$-edge between two vertices $u$ and $v$ at time $t$, denoted $\Gamma_\gamma(t,u,v)$, is the set $\{(t',\{u,v\}) \mid t' \in \intv{t,t+\gamma-1}\}$.
We say that a $\gamma$-edge $\Gamma$ contains a temporal vertex $(t,u)$ if there exists a vertex $v \in V$ such that $(t,\{u,v\}) \in \Gamma$.
We say that two $\gamma$-edges are independent if there is no temporal vertex that is contained in both of them.
A $\gamma$-matching $\mathcal{M}$ of a link stream $L$ is a set of pairwise independent $\gamma$-edges.
We say that a $\gamma$-edge $\Gamma$ is incident with a vertex $u \in V$ if there exist a vertex $v \in V$ and an integer $t \in T$ such that $\Gamma = \Gamma_\gamma(t,u,v)$.
We say that an edge $e \in E$ is in a $\gamma$-matching $\mathcal{M}$ if there exists $\Gamma \in \mathcal{M}$ such that $e \in \Gamma$.

We focus on the following problem.

\probl
{\textsc{$\gamma$-matching}}
{A link stream $L$ and an integer $k$.}
{A $\gamma$-matching of $L$ of size $k$ or a correct answer that such a set does not exist.}

\begin{theorem}
\textsc{$\gamma$-matching} is NP-hard for $\gamma>1$.
\end{theorem}

\begin{proof}
  We prove the NP-completeness of the decision version of \textsc{$\gamma$-matching}
  by a reduction from \textsc{3-Sat}, that is well known to be NP-complete.
  Let $\varphi$ be a formula with $n$ variables $x_1, \ldots, x_n$ and $m$ clauses $C_0, \ldots, C_{m-1}$ such that each clauses is of size at most $3$.
  Without loss of generality, we assume that a clause does not contain twice the same variable.
  We call $X$ the set containing the $n$ variables and $\mathcal{C}$ the set containing the $m$ clauses.

  We define the link stream $L=(T,V,E)$ in the following way:

  \begin{itemize}
  \item $T = \intv{0,(m+1)\gamma  -1}$.
  \item   $V = \{x^{-}, x^{=}, x^{+} \mid x \in X\} \cup \{ x^{++}_t,  x^{--}_t \mid x \in X, t \in \intv{0,m-1}\} \cup \{c\}$
  \item $E = E_{var} \cup E_{cla}$ where:
  \begin{eqnarray*}
    E_{var} &=& \{(t,\{x^{=},x^{+}\}), (t,\{x^{=},x^{-}\}) \mid  t \in \intv{0,(m+1)\gamma  -1}, x \in X\} \\
        &\cup& \{(t,\{x^{+},x^{++}_{i}\}) \mid t \in \intv{i\gamma+1, (i+1)\gamma}, i \in \intv{0,m-1}, x \in X \} \\
        &\cup& \{(t,\{x^{-},x^{--}_{i}\}) \mid t \in \intv{i\gamma+1, (i+1)\gamma}, i \in \intv{0,m-1}, x \in X \} \\
    E_{cla} & = & \{(t,\{c,x^{++}_i\}) \mid t \in \intv{i\gamma+1, (i+1)\gamma}, i \in \intv{0,m-1},  x \in X,\\ &&\qquad x \mbox{ appears positively in } C_{i}\}\\
        & \cup & \{(t,\{c,x^{--}_i\}) \mid  t \in \intv{i\gamma+1, (i+1)\gamma}, i \in \intv{0,m-1}, x \in X,\\ &&\qquad x \mbox{ appears negatively in } C_{i}\}.\\
  \end{eqnarray*}
  \end{itemize}

We depict in Figure~\ref{fig:nphard} the link stream build for $\gamma=3$ and $\varphi=(w \vee \overline{x} \vee y)\wedge (w \vee x \vee \overline{z})$.

We show that there is an assignments of the variables that satisfies $\varphi$ if and only if $L$ contains a $\gamma$-matching of size $(2m+1)n+m$.

\begin{figure}
  \centering
  \begin{tikzpicture}

    \def \dec {-0.5}
    \def \x {0*\dec}
    \node at (-1,0*\dec+\x) {\scriptsize{$w^{--}_1$}};
    \node at (-1,1*\dec+\x) {\scriptsize{$w^{--}_0$}};
    \node at (-1,2*\dec+\x) {\scriptsize{$w^-$}};
    \node at (-1,3*\dec+\x) {\scriptsize{$w^=$}};
    \node at (-1,4*\dec+\x) {\scriptsize{$w^+$}};
    \node at (-1,5*\dec+\x) {\scriptsize{$w^{++}_0$}};
    \node at (-1,6*\dec+\x) {\scriptsize{$w^{++}_1$}};

    \def \x {8*\dec}
    \node at (-1,0*\dec+\x) {\scriptsize{$x^{--}_1$}};
    \node at (-1,1*\dec+\x) {\scriptsize{$x^{--}_0$}};
    \node at (-1,2*\dec+\x) {\scriptsize{$x^-$}};
    \node at (-1,3*\dec+\x) {\scriptsize{$x^=$}};
    \node at (-1,4*\dec+\x) {\scriptsize{$x^+$}};
    \node at (-1,5*\dec+\x) {\scriptsize{$x^{++}_0$}};
    \node at (-1,6*\dec+\x) {\scriptsize{$x^{++}_1$}};

    \def \x {16*\dec}
    \node at (-1,0*\dec+\x) {\scriptsize{$y^{--}_1$}};
    \node at (-1,1*\dec+\x) {\scriptsize{$y^{--}_0$}};
    \node at (-1,2*\dec+\x) {\scriptsize{$y^-$}};
    \node at (-1,3*\dec+\x) {\scriptsize{$y^=$}};
    \node at (-1,4*\dec+\x) {\scriptsize{$y^+$}};
    \node at (-1,5*\dec+\x) {\scriptsize{$y^{++}_0$}};
    \node at (-1,6*\dec+\x) {\scriptsize{$y^{++}_1$}};

    \def \x {24*\dec}
    \node at (-1,0*\dec+\x) {\scriptsize{$z^{--}_1$}};
    \node at (-1,1*\dec+\x) {\scriptsize{$z^{--}_0$}};
    \node at (-1,2*\dec+\x) {\scriptsize{$z^-$}};
    \node at (-1,3*\dec+\x) {\scriptsize{$z^=$}};
    \node at (-1,4*\dec+\x) {\scriptsize{$z^+$}};
    \node at (-1,5*\dec+\x) {\scriptsize{$z^{++}_0$}};
    \node at (-1,6*\dec+\x) {\scriptsize{$z^{++}_1$}};

    \node at (-1,-2*\dec) {\scriptsize{$c$}};

    \draw[dotted] (-0.5,-2*\dec) -- (8.5, -2*\dec);

    \foreach \y in {0,1,2,3}
    {
      \foreach \z in {0,...,6}
      {
        \draw[dotted] (-0.5, \y *8*\dec+\z*\dec) -- (8.5, \y *8*\dec+\z*\dec);
      }

    }
    \foreach \x in {0,...,8}
    {
      \node at (\x,1-2*\dec) {$\x$};
      \node at (\x, -2*\dec) {$\bullet$};

      \draw[dotted] (\x,0.5-2*\dec) -- (\x,30*\dec-0.5);
              
      \foreach \y in {0,1,2,3}
      {
        \node at (\x, \y *8*\dec+3*\dec) {$\bullet$};
        \node at (\x, \y *8*\dec+2*\dec) {$\bullet$};
        \node at (\x, \y *8*\dec+4*\dec) {$\bullet$};
        \draw (\x, \y *8*\dec+4*\dec) -- (\x, \y *8*\dec+2*\dec);
      }
    }
    \foreach \x in {1,...,3}
    {
      \foreach \y in {0,1,2,3}
      {
        \node at (\x, \y *8*\dec+5*\dec) {$\bullet$};
        \node at (\x, \y *8*\dec+1*\dec) {$\bullet$};
        \draw (\x, \y *8*\dec+5*\dec) -- (\x, \y *8*\dec+1*\dec);
        
      }
    }
    \foreach \x in {4,5,6}
    {
      \foreach \y in {0,1,2,3}
      {
        \node at (\x, \y *8*\dec+6*\dec) {$\bullet$};
        \node at (\x, \y *8*\dec+0*\dec) {$\bullet$};
        \draw (\x, \y *8*\dec+6*\dec) -- (\x, \y *8*\dec+0*\dec);
        
      }
    }

    \foreach \x in {1,2,3}
    {
      \node[draw, circle] at (\x,0*8*\dec+5*\dec) {}; 
      \node[draw, circle] at (\x,1*8*\dec+1*\dec) {}; 
      \node[draw, circle] at (\x,2*8*\dec+5*\dec) {}; 
    }

        \foreach \x in {4,5,6}
    {
      \node[draw, circle] at (\x,0*8*\dec+6*\dec) {}; 
      \node[draw, circle] at (\x,1*8*\dec+6*\dec) {}; 
      \node[draw, circle] at (\x,3*8*\dec+0*\dec) {}; 
    }

\end{tikzpicture}

\caption{The constructed linkstream $L$ when $\varphi=(w \vee \overline{x} \vee y)\wedge (w \vee x \vee \overline{z})$ and $\gamma = 3$.
  Here $T = \intv{0,8}$ and the edge $(t,\{u,v\})$ of $L$ are depicted by an edge following the vertical line corresponding to time $t$ going from the horizontal line corresponding to $u$ to the horizontal line corresponding to $v$.
For readability, the edges incident with $c$ are not drown. Instead, we have circled the vertices that are neighbors of $c$ at each specific time.
}
\label{fig:nphard}
\end{figure}
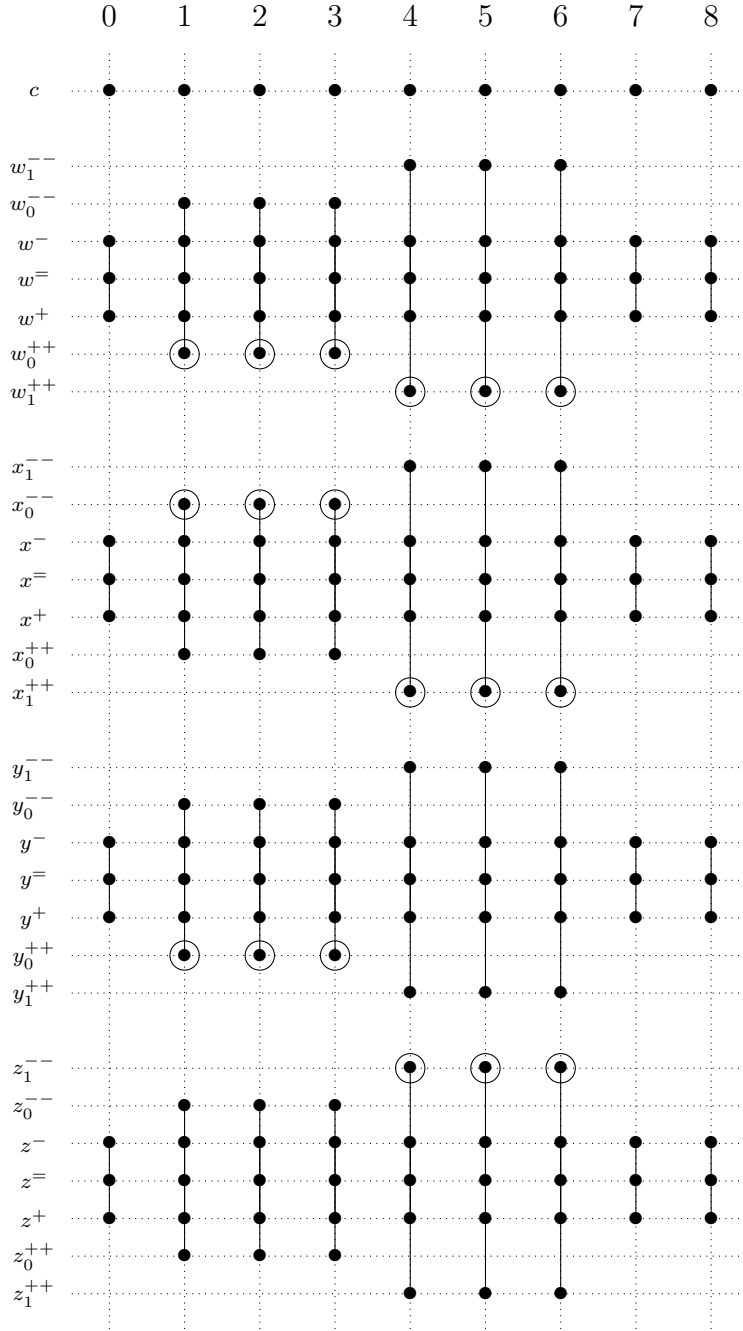
  
  Intuitively, the edge between  $(0,\{x^{=},x^{+}\})$ and   $(0,\{x^{=},x^{-}\})$ that is in the requested $\gamma$-matching 
  determines if the variable $x$ is set to true or false.
  Moreover, the size of the requested $\gamma$-matching will ensure that if
  the edge $(0,\{x^{=},x^{+}\})$ (resp. $(0,\{x^{=},x^{-}\})$) is in the $\gamma$-matching,
  then every edge $(t,\{x^{=},x^{+}\})$  (resp. $(t,\{x^{=},x^{-}\})$), $t \in \intv{0,(m+1)\gamma  -1}$
  and every edge $(t,\{x^{-},x^{--}_i\})$ (resp. $(t,\{x^{+},x^{++}_i\})$), $t \in \intv{1,m\gamma}\}$, $i = \left\lfloor \frac{t-1}{\gamma} \right\rfloor$,
  are in the $\gamma$-matching as well.
  Finally, during the time interval $\intv{i\gamma+1,(i+1) \gamma}$, we will certify that the clause $C_i$ is satisfied.

  First assume that $\varphi$ is satisfiable.
  Let $\psi$ be a satisfying assignment of $\varphi$ and let $\chi: \mathcal{C} \rightarrow V$ be a function that, for each clause $C_i$, $i \in \intv{0,m-1}$, arbitrary chooses a variable $x \in X$, such that the assignment of $x$ given by $\psi$ satisfies $C_i$, and returns $x^{++}_i$ (resp. $x^{--}_i$) if $\psi(x) = \mathtt{true}$ (resp. $\psi(x) = \mathtt{false}$).
  Let
  \begin{eqnarray*}
    \mathcal{M} &=& \{\Gamma_\gamma(i\cdot\gamma,x^=,x^+) \mid x \in X, \psi(x) = \mathtt{true}, i \in \intv{0,m}\}\\ 
    &\cup&  \{\Gamma_\gamma(i\cdot\gamma,x^=,x^-) \mid x \in X, \psi(x) = \mathtt{false}, i \in \intv{0,m}\}\\
    &\cup&  \{\Gamma_\gamma(i\cdot\gamma+1, x^{-},x^{--}_i) \mid x \in X, \psi(x) = \mathtt{true}, i \in \intv{1,m}\}\\
    &\cup&  \{\Gamma_\gamma(i\cdot\gamma+1, x^{+},x^{++}_i) \mid x \in X, \psi(x) = \mathtt{false}, i \in \intv{1,m}\}\\
    &\cup&  \{\Gamma_\gamma(i\cdot\gamma+1, c,\chi(C_t)) \mid i \in \intv{1,m}\}.\\
  \end{eqnarray*}
  One can  verify that $\mathcal{M}$ is a $\gamma$-matching of $L$ of size $(2m+1)n+m$.

Assume now that $L$ contains a $\gamma$-matching $\mathcal{M}$ of size $(2m+1)n+m$.
We use several claim in order to construct a satisfying assignment of $\varphi$.
\begin{claim}
  \label{claim1}
    For each $x \in X$,  $\mathcal{M}$ contains at most $m+1$ $\gamma$-edges incident with $x^{+}$ (resp. $x^-$).
  \end{claim}
  \begin{proof}
    This result follows from the fact that $T = \intv{0,(m+1)\gamma  -1}$ is of size $(m+1)\gamma$ and so, cannot be divided into $m+2$ pairwise disjoint sets of size $\gamma$.
  \end{proof}
  
  \begin{claim}
      \label{claim2}
      Given $x \in X$, if $\Gamma_\gamma(0,x^=,x^+) \not \in \mathcal{M}$ (resp. $\Gamma_\gamma(0,x^=,x^-) \not \in \mathcal{M}$), then $\mathcal{M}$ contains at most $m$ $\gamma$-edges
      incident with $x^{+}$ (resp. $x^{-}$).
\end{claim}
\begin{proof}
  As $\Gamma_\gamma(0,x^=,x^+)$ is the only $\gamma$-edge of $L$ that contains the edge $e_x^0=(0,\{x^=,x^+\})$, this implies that the edge $e_x^0$ is not contained in any $\gamma$-edge of $\mathcal{M}$.
  So the $\gamma$-edge of $\mathcal{M}$ that are incident with $x^+$ are constraint to exist in the time interval $I=\intv{1,(m+1)\gamma  -1}$ that is of size $(m+1)\gamma-1$.
  Thus, $I$ cannot be divided in $m+1$ pairwise disjoint sets of size $\gamma$. 
  The claim follows.
\end{proof}

\begin{claim}
  \label{claim4}
  $\mathcal{M}$ contains exactly $m$ $\gamma$-edges incident with $c$
  and contains exactly $2m+1$ $\gamma$-edges incident with $x^+$ or $x^-$, for each $x \in X$.
\end{claim}
\begin{proof}
  As $\mathcal{M}$ is a $\gamma$-matching, then for each $x \in X$,   $\Gamma_\gamma(0,x^=,x^+) \not \in \mathcal{M}$ or $\Gamma_\gamma(0,x^=,x^-) \not \in \mathcal{M}$. So Claim~\ref{claim1} and Claim~\ref{claim2} implie that for each $x \in X$,  $\mathcal{M}$ contains at most $2m+1$ $\gamma$-edges incident with $x^+$ or $x^-$.
  Moreover, by construction $\mathcal{M}$ can contains at most $m$ $\gamma$-edges incident with $c$ and $L$ does not contains any edge of the form
  $(t,\{x^+,y^-\})$,   $(t,\{x^+,y^+\})$,   $(t,\{x^-,y^-\})$,   $(t,\{x^+,c\})$, or $(t,\{x^-,c\})$, for any $x,y \in X$.
  Thus the budget is tight.
  The claim follows.
\end{proof}

Note that by construction, if $\mathcal{M}$ contains a $\gamma$-edge incident with $x^{++}_i$ for some $x \in X$ and $i \in \intv{0,m-1}$, then this $\gamma$-edge has to be either $\Gamma_\gamma(i\gamma+1,c,x^{++}_i)$ or $\Gamma_\gamma(i\gamma+1,x^+,x^{++}_i)$.
Moreover Claim~\ref{claim3} give us some information in the case where $\Gamma_\gamma(i\gamma+1,x^+,x^{++}_i) \in \mathcal{M}$
  \begin{claim}
      \label{claim3}
    Given $x \in X$, if $\mathcal{M}$ contains a $\gamma$-edge $\Gamma_\gamma(i\gamma+1, x^+,x^{++}_i)$ (resp. $\Gamma_\gamma(i\gamma+1, x^-,x^{--}_i)$) for some $i \in \intv{0,m-1}$, then  $\mathcal{M}$ contains at most $m$ $\gamma$-edges incident with $x^+$ (resp. $x^-$).
  \end{claim}
  \begin{proof}
    Let $x \in X$ and let $i$ be the first value such that $\Gamma_\gamma(i\gamma+1,x^+,x^{++}_i) \in \mathcal{M}$.
    As $\gamma$ does not divide $i\gamma+1$, this implies that, in the interval $\intv{0,i\gamma}$, at least one edge $e_x^t=(t,\{x^=,x^+\})$, $t \in \intv{0,i\gamma}$ is not in $\mathcal{M}$.
    So the $\gamma$-edges of $\mathcal{M}$ that are incident with $x^+$ are constraint to exist in the time interval $I=T\setminus{t}$ that is of size $(m+1)\gamma-1$.
      Thus, $I$ cannot be divided in $m+1$ pairwise disjoint sets of size $\gamma$. 
  The claim follows.
\end{proof}

Let $x \in X$.
Using Claim~\ref{claim4}, we know that 
$\mathcal{M}$ contains exactly $2m+1$ $\gamma$-edges incident with $x^+$ or $x^-$.
By the pigeonhole principle, we know that for $x^+$ or $x^-$, say $x^+$, $\mathcal{M}$ contains exactly $m+1$ $\gamma$-edges incident with $x^+$.
By Claim~\ref{claim3}, this implies that $\{\Gamma_\gamma(i\gamma,x^=,x^+) \mid i \in \intv{0,m}\} \subseteq \mathcal{M}$.
Thus, as $\mathcal{M}$ is a $\gamma$-matching that contains $m$ $\gamma$-edges incident with $x^-$, this also implies that $\{\Gamma_\gamma(i\gamma+1,x^-,x^{--}_i) \mid i \in \intv{0,m-1}\} \subseteq \mathcal{M}$.

For each variable $x \in X$, we set $x$ to \texttt{true} (resp. \texttt{false}) if $\Gamma_\gamma(0,x^=,x^+) \in \mathcal{M}$ (resp. $\Gamma_\gamma(0,x^=,x^-) \in \mathcal{M}$).
Let $\varphi$ be the so obtained assignment.
Let $i \in \intv{0,m-1}$.
We know that there exists $x \in X$ such that either $\Gamma_\gamma(i\gamma+1,c,x^{++}_i) \in \mathcal{M}$ or  $\Gamma_\gamma(i\gamma+1,c,x^{--}_i) \in \mathcal{M}$.
Let fix this $x \in X$ and assume that $\Gamma_\gamma(i\gamma+1,c,x^{++}_i) \in \mathcal{M}$, meaning that $x$ appears positively in $C_i$.
This implies that $\Gamma_\gamma(i\gamma+1,x^+,x^{++})\not \in \mathcal{M}$, so that  $\{\Gamma_\gamma(i'\gamma,x^=,x^+) \mid i' \in \intv{0,m}\} \subseteq \mathcal{M}$.
Thus, $x$ is set to \texttt{true} by $\varphi$ and $x$ satisfies $C_i$.
This concludes the proof.
\end{proof}

\section{Approximation algorithm}\label{sec:aa}
In classical graph theory, it is folklore that any maximal matching is also a $2$-approximation of a maximum matching, see e.g.~\cite[Exercice 35.4]{cormen}.
Fortunately enough, it is roughly the same situation with link streams.
Precisely, in this section, we adopt the greedy approach --finding a maximal $\gamma$-matching-- in order to provide a $2-$approximation algorithm for \textsc{$\gamma$-matching}.

Let $L = (T,V,E)$ be a link stream.
Let $\mathcal{P}$ be the set of all $\gamma$-edges of $L$.
Note that these $\gamma$-edges are not independent from each other, on the contrary, they highly overlap.
Let  $\preceq$ be an arbitrary  total ordering on the elements of $\mathcal{P}$ such that
given for any two elements of $\mathcal{P}$, $\Gamma_1 = \Gamma_\gamma(t_1,u_1,v_1)$ and  $\Gamma_2 = \Gamma_\gamma(t_2,u_2,v_2)$ such that $t_1 < t_2$, we have $\Gamma_1 \preceq \Gamma_2$.

We denote by $\mathcal{A}$ the following greedy algorithm. 
The algorithm starts with $\mathcal{M} = \emptyset$, $\mathcal{Q} = \mathcal{P}$, and a function $\rho: V\times T \rightarrow \{0,1\}$ such that for each $(t,v) \in T \times V$, $\rho(t,v) = 0$.
The purpose of $\rho$ is to keep track of the temporal vertices that are contained in a $\gamma$-edge of $\mathcal{M}$.
As long as $\mathcal{Q}$ is not empty, the algorithm selects $\Gamma$, the $\gamma$-edge of $\mathcal{Q}$ that is minimum for $\preceq$, and removes it from $\mathcal{Q}$.
Let $K$ be the set of the $2\gamma$ temporal vertices that are contained in $\Gamma$.
If, for each $(t,v) \in K$, $\rho(t,v) = 0$, then the algorithm adds $\Gamma$ to $\mathcal{M}$, otherwise it does nothing at this step.
For each $(t,v) \in K$, it sets $\rho(t,v)$ to $1$ and repeats.
If  $\mathcal{Q} = \emptyset$, it returns $\mathcal{M}$.

As $\mathcal{P}$ can be determined in a sorted way in time $\mathcal{O}(m)$, this algorithm runs in time $\mathcal{O}(n\tau+m)$, where $\tau = |T|$, $n = |V|$, $m =|E|$, and where $\gamma$ is a constant hidden in the $\mathcal{O}$.

Given a $\gamma$-matching $\mathcal{M}$, we define the \emph{bottom temporal vertices} of $\mathcal{M}$, denoted by $\mathtt{bot}(\mathcal{M})$, as the set 
$\{(t+\gamma-1,u), (t+\gamma-1,v) \mid \Gamma_\gamma(t,u,v) \in \mathcal{M}\}$. Lemma~\ref{th:bot} shows the crucial role of the bottom temporal vertices of the matchings returned by $\mathcal{A}$.

\begin{lemma}
  \label{th:bot}
  Let $\gamma$ be a positive integer, let $L$ be a link stream, and let $\mathcal{M}$ be a $\gamma$-matching returned by $\mathcal{A}$ when applied to $L$. If $\mathcal{M}'$ is a $\gamma$-matching of $L$, then every $\gamma$-edge of $\mathcal{M'}$ contains, at least, one temporal vertex of $\mathtt{bot}(\mathcal{M})$.
\end{lemma}

\begin{proof}
  First, note that any $\gamma$-edge of $\mathcal{M}$ contains two temporal vertices of $\mathtt{bot}(\mathcal{M})$, and so, at least one.
  Let $\Gamma'$ be a $\gamma$-edge of $\mathcal{M'}$ that is not in $\mathcal{M}$.
  Let $\mathcal{M}^* \subseteq \mathcal{M}$ be the set of every $\gamma$-edge $\Gamma^*$ of $\mathcal{M}$ such that there exists a temporal vertex $(t,u)$ that is contained in both $\Gamma'$ and $\Gamma^*$.
  Assume that $\Gamma' = \Gamma_\gamma(t,u,v)$.
  If there exists $\Gamma^* \in \mathcal{M}$ such that $\Gamma^* = \Gamma_\gamma(t',u,v')$ and $t' \leq t$, then we have that $(t'+\gamma-1,u) \in \mathtt{bot}(\mathcal{M})$ is contained in $\Gamma'$.
  Otherwise, we have that for each  $\Gamma^* \in \mathcal{M}^*$ such that $\Gamma^* = \Gamma_\gamma(t',u,v')$, $t' > t$.
  This is not possible by construction of $\mathcal{A}$.
  This concludes the proof.
\end{proof}

Lemma~\ref{th:bot} plays a cornerstone role in the proof of subsequent Theorem~\ref{maintheorem}.
As a byproduct, we also obtain the following result.
\begin{corollary}
  \label{thm:approx}
$\mathcal{A}$ is a $2$-approximation of the \textsc{$\gamma$-matching} problem.
\end{corollary}
\begin{proof}
  Let $L$ be the input link stream.
  Let $\mathcal{M}$ be a solution returned by the algorithm $\mathcal{A}$ when applied to $L$, and let $\mathcal{M}'$ be a $\gamma$-matching of $L$.
  As
  $|\mathtt{bot}(\mathcal{M})| = 2|\mathcal{M}|$,
  two $\gamma$-edges of $\mathcal{M}'$ cannot contains the same temporal vertex, and,
  by Lemma~\ref{th:bot}, every $\gamma$-edge of $\mathcal{M}'$ contains at least one element of $\mathtt{bot}(\mathcal{M})$, 
we obtained that $|\mathcal{M}'| \leq 2|\mathcal{M}|$.
\end{proof}

\section{Kernelization algorithm}\label{sec:ka}
Problem \textsc{$\gamma$-matching} has a \textit{kernel} if there exist a computable function $f:\mathbb{N}\rightarrow\mathbb{N}$ and a polynomial time algorithm $\ca$ which takes as input an instance $(L,k)$ of \textsc{$\gamma$-matching} and produces an instance $(L',k')$ such that: $k'\leq k$; $|L'|\leq f(k)$; and $(L',k')$ yields a positive answer for \textsc{$\gamma$-matching} if and only if $(L,k)$ yields a positive answer for \textsc{$\gamma$-matching}.
In this case, algorithm $\ca$ is called a \textit{kernelization algorithm} for \textsc{$\gamma$-matching}~\cite{DF99}.

We now show a kernelization algorithm for \textsc{$\gamma$-matching} by a direct pruning process based on Lemma~\ref{th:bot}.
For convenience, let us say that a $\gamma$-edge $\Gamma$ is incident to a temporal vertex $(t,u)$ when there exists vertex $v\neq u$ such that $(t,\{u,v\})\in\Gamma$.
The main idea is as follows.
First, we compute the set $S$ of all bottom temporal vertices of a $\gamma$-matching produced by previously defined algorithm $\mathcal A$.
Then, we prune the original instance by only keeping edges that belong to a $\gamma$-edge incident to a temporal vertex of $S$.
More precisely, we prove the following result.
\begin{theorem}\label{maintheorem}
  There exists a polynomial-time algorithm that for each instance $(L,k)$, either returns a true instance which correctly correspond to the fact that $L$ contains a $\gamma$-matching of size $k$, or returns an equivalence instance $(L',k)$ such that the number of edges of $L'$ is $2(k-1)(2k-1)\gamma^2$.
\end{theorem}
\begin{proof}
  Let $L = (T,V,E)$ be a link stream and $k$ be an integer.
  We first run the algorithm $\mathcal{A}$ on $L$.
  Let $\mathcal{M}$ be the $\gamma$-matching outputed by the algorithm and let $\ell = |\mathcal{M}|$.
  If $\ell\geq k$, then we already have a solution and then return a true instance.
  If $\ell< \frac{k}{2}$, then, by Corollary~\ref{thm:approx}, we know that the instance does not contains a $\gamma$-matching of size $k$, and then we return a false instance.
  We now assume that $\frac{k}{2} \leq \ell < k$.

  Lemma~\ref{th:bot} justifies that we are now focusing on the temporal vertices of $\mathtt{bot}(\mathcal{M})$ in order to find the requested kernel.
  We construct a set  $\mathcal{P}$ of $\gamma$-edges
  and we show that any edge $e$, that is not in a $\gamma$-edge of $\mathcal{P}$, is useless when looking for a $\gamma$-matching of size $k$.
  For each $(t,u) \in \mathtt{bot}(\mathcal{M})$, and for each $t'$ such that $\max(0,t-\gamma+1) \leq t' \leq t$, we consider the set $\mathcal{S}(t',u)$ of every $\gamma$-edge, existing in $L$, with the form $\Gamma_\gamma(t',u,v)$ with $v \in V$.
  If the set  $\mathcal{S}(t',u)$ is of size at most $2k-1$, we add every element of  $\mathcal{S}(t',u)$ to $\mathcal{P}$.
  Otherwise, we select $2k-1$ elements of $\mathcal{S}(t',u)$ that we add to $\mathcal{P}$.
  In both cases, we denote by $\mathcal{S'}(t',u)$ the set of elements of  $\mathcal{S}(t',u)$ that we have added to $\mathcal{P}$.
  This finish the construction of $\mathcal{P}$.
  As   $|\mathtt{bot}(\mathcal{M})| = 2\ell$ and for each element of $\mathtt{bot}(\mathcal{M})$ we have added at most $(2k-1)\gamma$ $\gamma$-edges to $\mathcal{P}$, we have that $|\mathcal{P}| \leq 2\ell (2k-1)\gamma \leq 2 (k-1)(2k-1)\gamma$.

  We now prove that if $L$ contains a $\gamma$-matching $\mathcal{M'}$ of size $k$, then it also contains a $\gamma$-matching $\mathcal{M''}$ of size $k$ such that $\mathcal{M''} \subseteq \mathcal{P}$.
  Let $\mathcal{M'}$ be a $\gamma$-matching of $L$ of size $k$ such that $p = |\mathcal{M'} \setminus \mathcal{P}|$ is minimum.
  We have to prove that $p=0$.
  Assume that $p \geq 1$. Let $\Gamma$ be a $\gamma$-edge in $\mathcal{M'} \setminus \mathcal{P}$.
  Let $(t,u)$ be a temporal vertex of $\mathtt{bot}(\mathcal{M})$ that is contained in $\Gamma$.
  We know by Lemma~\ref{th:bot} that this temporal vertex exists.
  Assume that $\Gamma = \Gamma_\gamma(t',u,v)$ for some $v \in V$ and some $t'$ such that $\max(0,t-\gamma+1) \leq t' \leq t$.
  As $\Gamma \not \in \mathcal{P}$, we have that $\Gamma \in \mathcal{S}(t',u) \setminus \mathcal{S'}(t',u)$, and so $|\mathcal{S'}(t',u)| = 2k-1$.
  Let $N_{\mathcal{S'}}(t',u)$ be the set of vertices $w$ of $V \setminus \{u\}$ such that a $\gamma$-edge of $\mathcal{S'}(t',u)$ is incident to $w$.
  As $\mathcal{M'} \setminus \{\Gamma\}$ is of size $k-1$, the $\gamma$-edges that it contains can be incident to at most $2k-2$ vertices. This means that there exists $w \in N_{\mathcal{S'}}(t',u)$ such that no $\gamma$-edge of $\mathcal{M'} \setminus \{\Gamma\}$ is incident to $w$.
  Thus  $(\mathcal{M'} \setminus \{\Gamma\}) \cup \{\Gamma_\gamma(t',u,w)\}$ is a $\gamma$-matching of size $k$.
  As $\Gamma \not \in \mathcal{P}$ and $\Gamma_\gamma(t',u,v) \in \mathcal{P}$, this contradicts the fact that $p$ is minimum.

  We now can define the link stream $L' = (T,V,E')$ such that $E' = \{e \in E \mid \exists \Gamma \in \mathcal{P} : e \in \Gamma\}$.
  As  $|\mathcal{P}| \leq 2 (k-1)(2k-1)\gamma$ and every element of $\mathcal{P}$ is a $\gamma$-edge, we have that
  $|E'| \leq 2(k-1)(2k-1)\gamma^2$.
  The theorem follows.
\end{proof}

\section{Experimental result}\label{sec:num}
For easy diffusion, both our $2$-approximation and kernelization algorithms are implemented in Java and JavaScript\footnote{The source code is available at\\\url{https://github.com/antoinedimitriroux/Temporal-Matching-in-Link-Streams}}.
Experiments are run on a standard laptop clocking at 3,1 Ghz with DDR3 16Go memory.

\subsection{Dataset}
We carried out our experiments on two main types of datasets:
those that are randomly generated\footnote{Direct link to the GUI of the generator:\\\url{https://antoinedimitriroux.github.io}};
and those that are collected from human activities.

\textit{Artificially generated link streams, and stress test:}\\
In order to generate random sets of link stream instances, we adopt the more realistic point of view of random geometric graphs, rather than the classically theoretic Erd\"os-R\'enyi model, as follows.
Let $\cs$ be a 2D Euclidian space.
We define a particle as a point in space $\cs$.
Every particle is given along with a radius representing the maximum communication distance it can have with another particle.
Thus, the particle together with its range define a disk in space $\cs$.
Let $\cp$ be a set of particles, given along with the same value of radius.
We partition set $\cp$ into $n$ parts, $\cp = P_1\cup P_2\cup\dots\cup P_n$, of roughly equal size.
We will construct a link stream $L=(T,V,E)$ as follows.
Let $V=\{P_1,P_2,\dots,P_n\}$.
At time zero, let $E_0=\{(0,\{P_i,P_j\})~|~\exists a\in P_i\land b\in P_j, \textit{ the distance between $a$ and $b$ is less than their radius}\}$.
In other words, if there are at least two particles under communication range $a\in P_i, b\in P_j$ of different groups $P_i\neq P_j$ (at time zero), then, we consider there is a (zero-timed) edge between $P_i$ and $P_j$.
Every particle has a velocity that is defined as follows.
First, the velocity of a particle at time $t$ is a fraction of the velocity at time $t-1$ of that particle (friction).
Second, all velocity vectors are subject to a small random additional factor (random walk factor modeling the wind condition).
Finally, we truncate every velocity vector in order to insure that the norm of the vector is lower than a given maximum particle speed (physical limits).
We then let the system evolve during a given laps of time, that we also refer to as $T$.
At every time instant $t\in T$, we define, similarly as before, the $t$-timed edge set
 $E_t=\{(t,\{P_i,P_j\})~|~\exists a\in P_i\land b\in P_j, \textit{ the distance between $a$ and $b$ is less than their radius}\}$.
Finally, we define our link stream as
$E=\bigcup_{t\in T}E_t$.

Roughly, increasing any of the three parameters which are defined by the particle radius, the cardinality of $\cp$, and the maximum particle speed, results in the same effect on the generated link stream, that is, to produce a dense link stream.
The generated data allows us to unit test our code, and especially to verify that approximation and kernelization runtime is sound on large input.
For instance, with inputs containing hundreds of thousand timed edges, our runtime is under ten seconds, cf.\ Fig.~\ref{fig_stresstest}.
\begin{figure}
  \centerline{\hfill\includegraphics[width=0.45\textwidth]{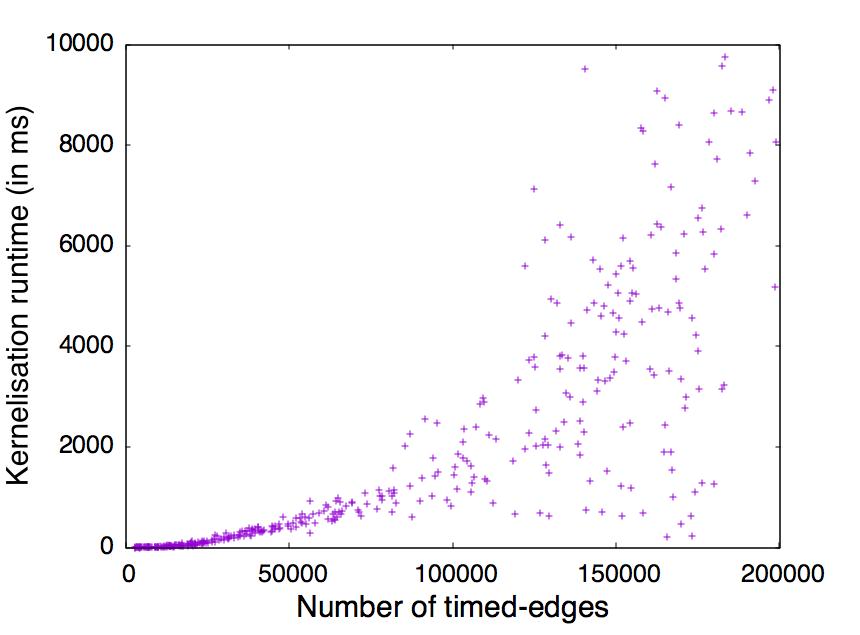}\hfill\hfill
                    \includegraphics[width=0.45\textwidth]{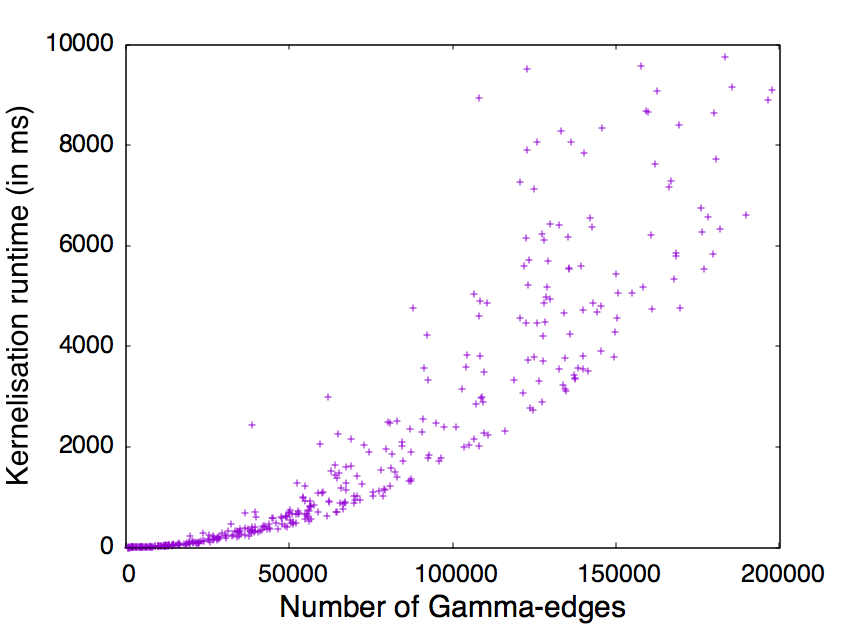}\hfill}
  \caption{\label{fig_stresstest} Stress-test on generated input, with $\gamma=5$.
  Though the number of $\gamma$-edges can be extremely high compared to timed edges (cf.\ overlapping $\gamma$-edges), we parametrised the generator so that they look the same in the left chart and the right chart.
  This is by no means a general property.}
\end{figure}
  
Interestingly, in the left chart of Fig.~\ref{fig_stresstest}, we note for some input with a large number of timed edges and no $\gamma$-edges that the runtime can be very quick.
We also discuss this phenomenon on real world dataset, in below Fig.~\ref{figure_enron_1day}~and~\ref{figure_rollernet_2hours}.
For this reason, it is much more interesting to examine the runtime as a function of the number of $\gamma$-edges of the input, cf.\ the right chart in Fig.~\ref{fig_stresstest}.
For easy eye-comparison, we tried to parametrise the generator in a way that most of the generated instances have roughly the same number of timed edges and $\gamma$-edges.
For instance, most instances with less than $100000$ timed edges also have roughly the same number of $\gamma$-edges.
However, the situation is more random for instances having between $100000$ and $200000$ timed edges.

\textit{Real world datasets, with cleaning methodology:}\\
We also confront the implementations of our algorithms to two particular link streams collected from real world graph data.
In one dataset the link stream is built from emailing information collected from the Enron company~\cite{KY04}.
The other dataset has been built by analysing a recording of $2\times80$ minute Rollerblade touring in Paris~\cite{TLBCDW09}.
Because of the long duration of these two experiments, the number of vertices (under two hundred persons in both cases) is negligible when comparing to the number of temporal vertices (nearly one million for Rollernet).
For a link stream $L=(T,V,E)$, we mostly compare $|T|$ with $|E|$ to get a glimpse on the density of the links.
For a complete view of $|T|$, $|V|$, and $|E|$ in the datasets, we refer the reader to Fig.~\ref{figure_enron_rollernet_generated_gamma2}.
Furthermore, we noticed with our raw datasets that the instants where some timed edge is present can be very sparse, leaving no chance for a $\gamma$-edge to exist as soon as $\gamma>1$.
For instance, in the Enron experiment (Fig.~\ref{figure_enron_1day}), we can see that there are no pair of employees who keep sharing $1$ mail per hour during $24$ hours, probably due to inactivity at night.
We will, for this reason, time-compress our raw datasets by the following process.
\begin{definition}[Data cleaning by time-compression]
\emph{For any link stream $L=(T,V,E)$ and for any $1<\delta<|T|$, we define the \textit{$\delta$-compression} $L_\delta=(T_\delta,V_\delta,E_\delta)$ as
$V_\delta=V$, $T_\delta=\llbracket \frac{\min T}{\delta}, \frac{\max T}{\delta} \rrbracket$, and
$$E_\delta=\{(t,\{u,v\})~|~\exists t'\in T: \delta t\leq t'<\delta(t+1) \land (t',\{u,v\})\in E\}.$$}
\end{definition}
\begin{figure}
\begin{center}
  \begin{scriptsize}
\begin{tabular}{|c|c|c|c|c|c|c|c}
\hline
$\delta$&$|T|$&$|E|$&$\gamma$&$|\gamma_E|$\\
\hline
\hline
$3600s=1h$&27300&21959&24&0\\
\hline
$7200s=2h$&13650&20962&12&0\\
\hline
$10800s=3h$&9100&20284&8&0\\
\hline
$14400s=4h$&6825&19732&6&16\\
\hline
$21600s=6h$&4550&19071&4&69\\
\hline
$28800s=8h$&3413&18402&3&335\\
\hline
$43200s=12h$&2275&17610&2&2667\\
\hline
\end{tabular}
  \end{scriptsize}
\caption{Enron dataset: number of timed edges and $\gamma$-edges after time-compression.
  Values are taken such that $\delta * \gamma = 24 hours$.
  In particular, we observe that Enron employee will not continuously share 1 mail per hour during 24 hours, since the company is closed at night.
  When compared to the number $|T|$ of time instants, the number $|V|$ of vertices is very small (under two hundred) and not presented here.}
\label{figure_enron_1day}
\end{center}
\end{figure}
\begin{figure}
\begin{center}
  \begin{scriptsize}
\begin{tabular}{|c|c|c|c|c|c|c|c}
\hline
$\delta$&$|T|$&$|E|$&$\gamma$&$|\gamma_E|$\\
\hline
\hline
$1s$&9977&403834&7200&0\\
\hline
$5s$&1996&127401&1240&0\\
\hline
$15$&666&77989&480&0\\
\hline
$30s$&333&60919&240&0\\
\hline
$60s=1m$&167&45469&120&0\\
\hline
$300s=5m$&34&22484&24&51\\
\hline
$600s=10m$&17&15808&12&357\\
\hline
$1200s=20m$&9&10735&6&1893\\
\hline
$1800s=30m$&6&8324&4&2745\\
\hline
$3600s=1h$&3&5000&2&3094\\
\hline
\end{tabular}
  \end{scriptsize}
\caption{Rollernet dataset: number of timed edges and $\gamma$-edges after time-compression.
  Values are taken such that $\delta * \gamma = 7200s = 2 hours$.
  In particular, we observe that every person in the Rollernet experiment has been away from another person for at least 1 minutes during 2 hours.
  When compared to the number $|T|$ of time instants, the number $|V|$ of vertices is very small (under on hundred) and not presented here.}
\label{figure_rollernet_2hours}
\end{center}
\end{figure}

Fig.~\ref{figure_enron_1day}~and~\ref{figure_rollernet_2hours} show that parameters $\delta$ and $\gamma$ have an important influence on the number of $\gamma-$edges. 
One can also notice from Fig.~\ref{figure_enron_1day}~and~\ref{figure_rollernet_2hours} that the compression process generally breaks down the number $|E|$ of timed edges in the dataset.
However, we stress that it is not necessarily the case: Fig.~\ref{explication_delta} exemplifies two different time-compressions of the same original link stream, respectively with $\delta = 3$ and $\delta = 4$, where it is possible to obtain more edges even if we have a larger $\delta$.
Nonetheless, the usual effect of time-compression is to drastically reduce the number of timed edges.
On Enron and Rollernet datasets, we need to ensure that time-compression does not result in empty inputs.
Luckily, for sensible values of $\delta$, e.g. $\frac{1}{2}$ week for Enron or $15$ minutes for Rollernet, there is still a large number of timed edges after $\delta$-compression, cf.\ Fig.~\ref{fig_timecompression}.
We give in the subsequent Fig.~\ref{figure_enron_rollernet_generated_gamma2}~and~\ref{fig_realtime} the runtime of our algorithm on random pieces of the two Enron and Rollernet datasets, where we observe that our runtime is very quick.
\begin{figure}
  \centerline{\includegraphics[width=.67\textwidth]{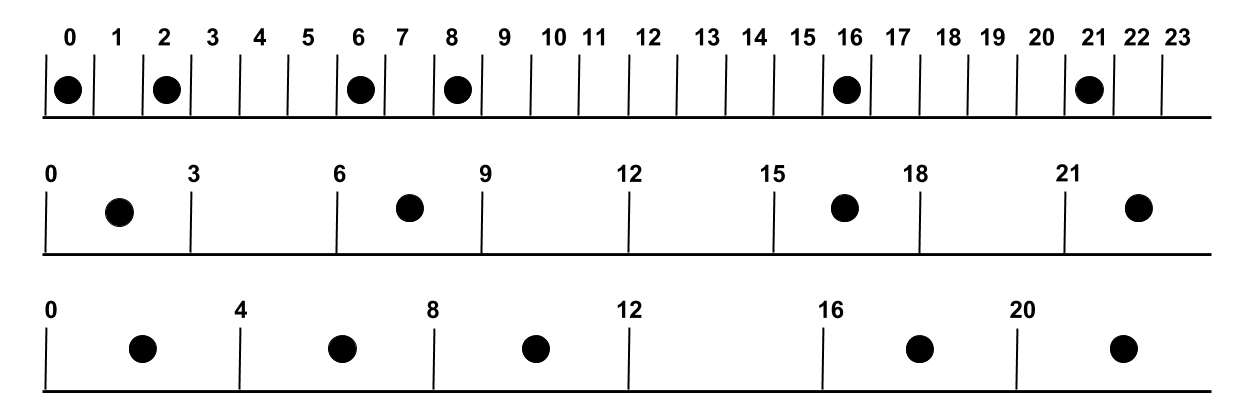}}
    \caption{\label{explication_delta} Time-compression with $\delta = 3$ and $\delta = 4$, resulting in $4$ and $5$ timed edges.}
\end{figure}
\begin{figure}
  \centerline{\hfill\includegraphics[width=0.45\textwidth]{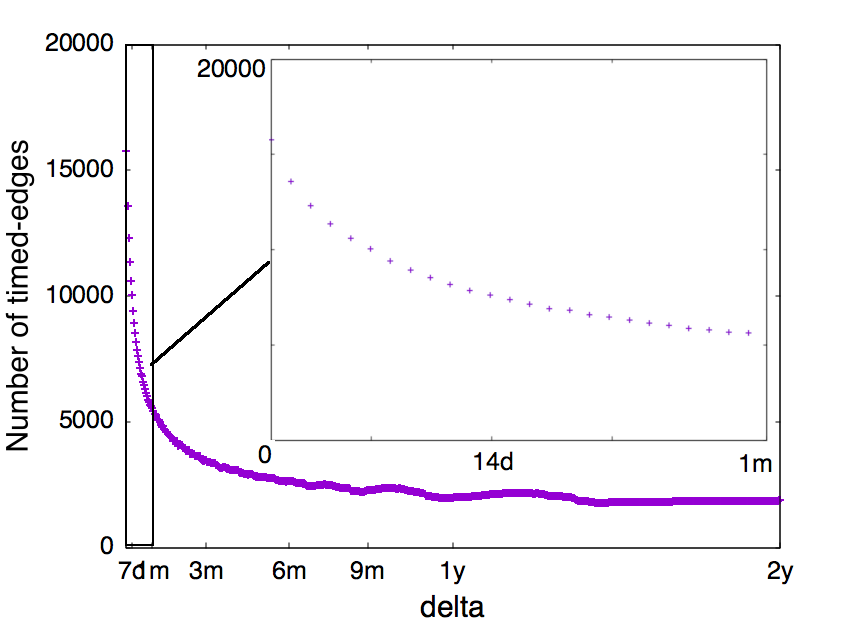}\hfill\hfill
                    \includegraphics[width=0.45\textwidth]{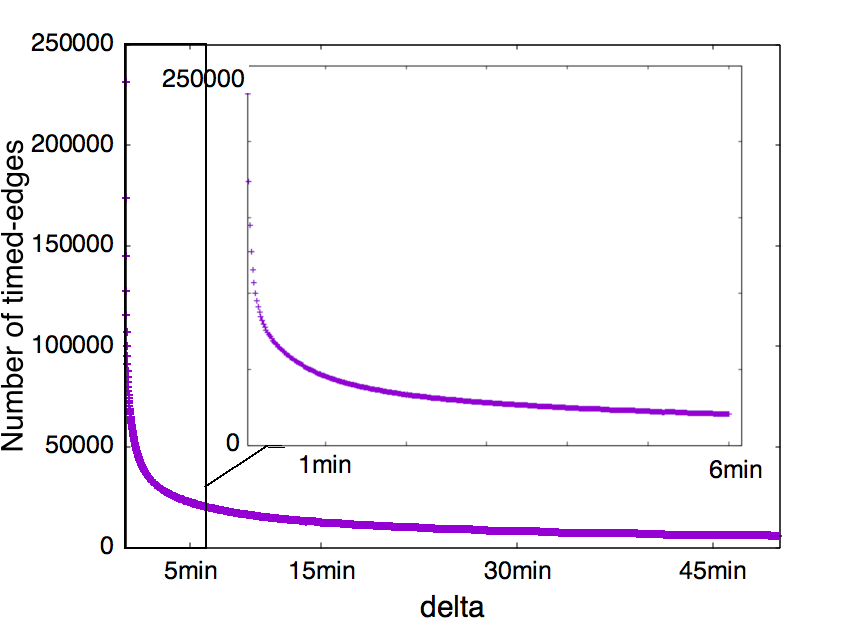}\hfill}
  \caption{\label{fig_timecompression} Enron dataset (left) and Rollernet dataset (right): remaining timed edges after $\delta$-compression.
  For instance, with $\delta=\frac{1}{2}$ week, the number of timed edges after $\delta$-compression Enron is over ten thousand.
  With $\delta=15$ minutes, the number of timed edges after $\delta$-compression Rollernet is also over ten thousand.
  We conclude that for sensible values of $\delta$, our time compression process does not break down the input to an empty instance.}
\end{figure}
\begin{figure}
\begin{center}
  \begin{tiny}
\begin{tabular}{|c|c|c|c|c|c|c|c|}
\hline
Link-Stream$_{\delta}$&$|V|$&$|T|$&$|E|$&$|\gamma_E|$&appr(s)&kern(s)&total(s)\\
\hline
\hline
$Enron_{1hour}$&150&27300&21959&1991&0.010&0.397&0.408\\
\hline
$Enron_{3hours}$&150&9100&20284&2695&0.018&0.694&0.696\\
\hline
$Enron_{1day}$&150&1138&16224&4416&0.007&1.76&1.774\\
\hline
$Enron_{3days}$&150&380&12868&4644&0.007&1.932&1.939\\
\hline
$Enron_{7days}$&150&163&10028&4812&0.005&1.173&1.179\\
\hline
$Enron_{30days}$&150&38&5573&2917&0.001&0.147&0.149\\
\hline
$Enron_{90days}$&150&13&3480&1650&6.6E-4&0.029&0.030\\
\hline
$Rollernet_{1min}$&61&167&45469&24009&0.098&1.696&1.794\\
\hline
$Rollernet_{2mins}$&61&84&33304&17089&0.047&0.561&0.609\\
\hline
$Rollernet_{5mins}$&61&34&22484&13346&0.018&0.140&0.158\\
\hline
$Rollernet_{15mins}$&61&12&12410&8544&0.005&0.044&0.050\\
\hline
$Rollernet_{30mins}$&61&6&8324&5979&0.005&0.032&0.038\\
\hline
$Rollernet_{1hour}$&61&3&5000&3094&0.001&0.007&0.008\\
\hline
$Generated$&10&50&684&83&4.4E-4&0.001&0.001\\
\hline
$Generated$&10&100&1384&136&4.8E-4&7.4E-4&0.001\\
\hline
$Generated$&10&200&2906&322&4.1E-4&0.002&0.002\\
\hline
$Generated$&20&50&2599&705&0.001&0.004&0.006\\
\hline
$Generated$&20&100&5326&1508&0.003&0.016&0.020\\
\hline
$Generated$&20&200&10667&2998&0.004&0.030&0.035\\
\hline
$Generated$&50&50&15842&6534&0.005&0.034&0.040\\
\hline
$Generated$&50&100&31113&12876&0.018&0.125&0.144\\
\hline
$Generated$&50&200&63032&26495&0.054&0.426&0.480\\
\hline
$Generated$&100&50&53665&32235&0.093&0.280&0.374\\
\hline
$Generated$&100&100&107524&65145&0.342&1.054&1.396\\
\hline
$Generated$&100&200&214728&130371&1.437&4.277&5.713\\
\hline
\end{tabular}
  \end{tiny}
\caption{Runtime required by our $2$-approximation algorithm, kernelization algorithm, and the total process.
  Values are taken for $\gamma=2$.
  The parameter $k$ for the kernelization algorithm is the size of the solution found by the $2$-approximation algorithm.
  The (rawly recorded) runtime is very quick, hence, probably subject to many noises.}
\label{figure_enron_rollernet_generated_gamma2}
\end{center}
\end{figure}
\begin{figure}
  \centerline{\hfill\includegraphics[width=0.45\textwidth]{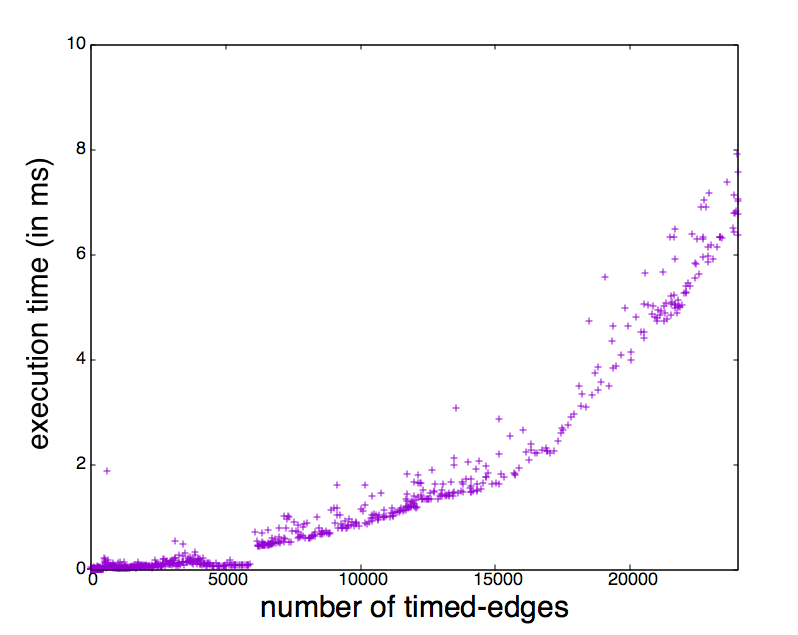}\hfill\hfill
                    \includegraphics[width=0.45\textwidth]{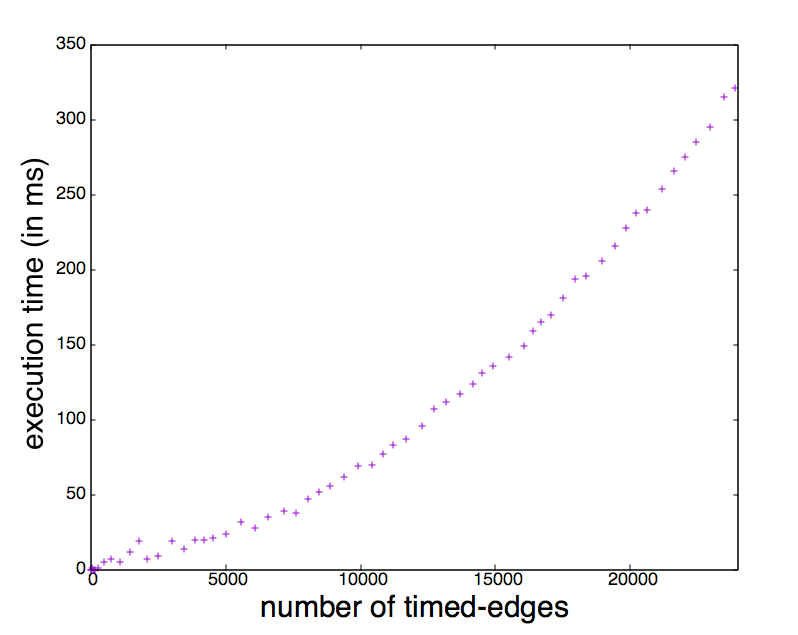}\hfill}
  \caption{\label{fig_realtime} Enron dataset (left) and Rollernet dataset (right): runtime of our algorithms in function of the number of timed edges, with $\gamma=2$.
  The parameter $k$ for kernelization is the size of the $2$-approximation result.
  The (rawly recorded) runtime is probably subject to many noises.
  We rather refer to Fig.~\ref{fig_stresstest} for evaluating performance.
  Each dot in current Fig.~\ref{fig_realtime} is obtained by first truncating the raw input with varying maximum value of time instants, then, $\delta$-compressing the so-obtained link streams with $\delta=100$.
  The only observation we make with this figure is that, on real world dataset, our combined runtime for both the $2$-approximation and the kernelization algorithms is very quick.}
\end{figure}

\subsection{Hypothesis}
We theorise three hypotheses.
For each hypothesis we run experiment on the above described datasets, and expose our results in the next Subsection~\ref{subsec_result}.
We discuss and conclude our numerical analysis in the subsequent and last Subsection~\ref{subsec_discussion} of current Section~\ref{sec:num}.

\textit{Hypothesis 1. Consistence of the formalism:}\\
We would like to verify that \textsc{$\gamma$-matching} is non-trivial on human values for $\delta$ and $\gamma$.
For instance, we suppose when mining emails that collaborating during a month at a rate of at least two emails per week (round trip) is sensibly human values.
When mining proximity records of $2\times 80$ minute Rollerblade touring Paris, we consider that collaborating during $80$ minutes at a rate of one visit every quarter hour (water/snack supplying) is sensibly human values.

\textit{Hypothesis 2. Kernelization quality:}\\
We reckon that, beside the hidden constants under the Landau notation, a space reduction from $O(n)$ to $O(k^2)$, when $k=O(\sqrt n)$ is also meaningless.
Unfortunately, when parameterizing by the size of the solution as we do in this paper, one very usually results in the situation where $k$ is numerically in the order of $\sqrt n$.
This is particularly true for our study of temporal matchings on Enron and Rollernet.
For this reason, our second hypothesis is that, in addition to a theoretical guarantee of reduction from $O(n)$ to $O(k^2)$ space complexity, solving \textsc{$\gamma$-matching} can numerically benefit from the kernelization algorithm described in Theorem~\ref{maintheorem}, at least on well-chosen and humanly sensible intervals of $\delta$ and $\gamma$.

\textit{Hypothesis 3. Approximation quality:}\\
We stress that \textsc{Matching} in a classical graph is polynomial.
Unfortunately, \textsc{$\gamma$-matching} in a link stream is NP-hard.
However, in practice, a lot of NP-complete problems are not difficult on datasets arising from human activities.
What's more, some such problems can be solved near-optimally by simple algorithms such as by a random or greedy approach, or a mix of both approaches, even on arbitrary inputs.
A popular example is \textsc{Coloring}~\cite{MR02}.
Accordingly, our last hypothesis is that, in practice, finding an optimal $\gamma$-matching need not to be difficult.
Moreover, we hypothesise that the greedy $2$-approximation described in Lemma~\ref{th:bot} can produce near-optimal $\gamma$-matching on real world dataset, as well as artificial datasets that mimic real word datasets.

\subsection{Result}
\label{subsec_result}
Both our $2$-approximation and kernelization algorithms are implemented and confronted to the above mentioned datasets.

\textit{Observations w.r.t.\ Hypothesis 1. Consistence of the formalism:}\\
Results are given in Fig.~\ref{fig_consistence}.
We observe 
on Enron dataset with $\delta\approx\frac{1}{2}$ week
that, after $\delta$-compression, the number of $\gamma$-edges for $\gamma$ varying from $2$ to $10$ is: over $500$ for $\gamma=10$; over $1000$ for $\gamma=6$; and over $4500$ for $\gamma=2$.
Moreover, we will see in the next paragraph that our numerical analysis can only find $\gamma$-matchings of size approximately one fifth of the previous number.
We also tried other techniques to improve the size of the $\gamma$-matchings but failed in finding substantial difference.
With the Enron dataset, we believe that \textsc{$\gamma$-matching} for $\gamma\approx 10$ is a tricky question when the time-compression rate is $\delta\approx\frac{1}{2}$ week.
These values translate the fact that any pair of collaborators in the $\gamma$-matching necessarily keep exchanging emails together continuously for one month, at a rate of at least one email per week and in average at least two emails every week.
\begin{figure}
  \centerline{\hfill\includegraphics[width=0.45\textwidth]{graphics/fig_a_enron_nombre_gamma_edges_combi.png}\hfill\hfill
                    \includegraphics[width=0.45\textwidth]{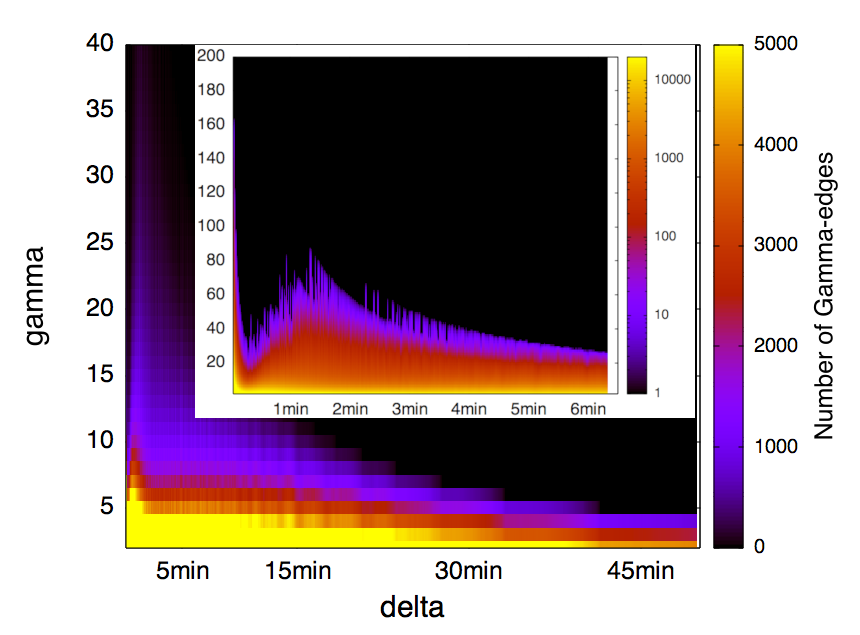}\hfill}

  \caption{\label{fig_consistence} Enron dataset (left) and Rollernet dataset (right): number of $\gamma$-edges after $\delta$-compression, for varying values of $\delta$ and $\gamma$.}
\end{figure}
\begin{figure}
  \centerline{\hfill\includegraphics[width=0.45\textwidth]{graphics/fig_c_enron_ratio_kernel_sur_gamma_L_combine.png}\hfill\hfill
                    \includegraphics[width=0.45\textwidth]{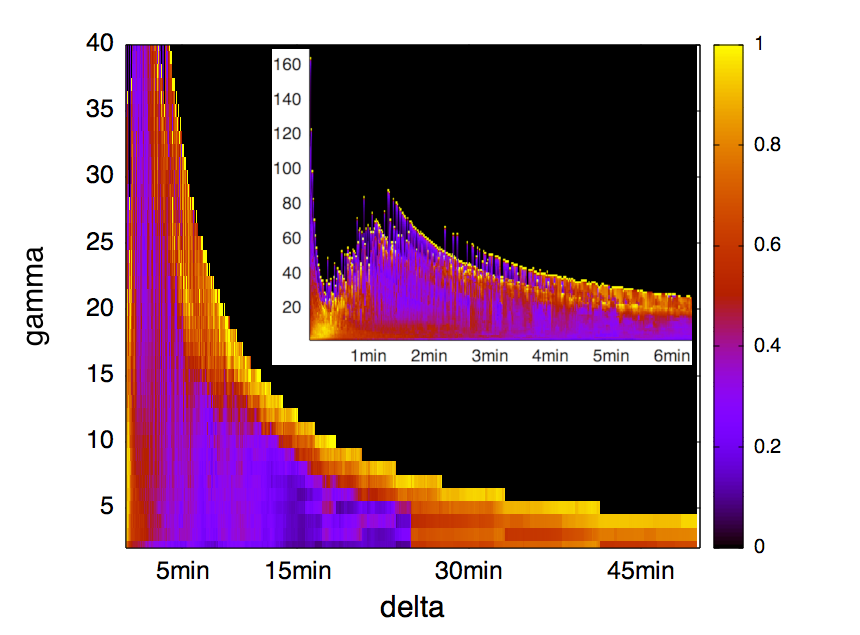}\hfill}

  \caption{\label{fig_kern_quality} Enron dataset (left) and Rollernet dataset (right): ratio obtained by dividing the number of $\gamma$-edges in the kernelization output by the number of $\gamma$-edges in the kernelization input (which is obtained after $\delta$-compression). Here, the darker is the better.
The parameter $k$ for the kernelization algorithm is the size of the solution found by the $2$-approximation algorithm.}
\end{figure}
\begin{figure}
  \centerline{\hfill\includegraphics[width=0.45\textwidth]{graphics/fig_e_enron_ratio_approx_sur_kernel_combi.png}\hfill\hfill
                    \includegraphics[width=0.45\textwidth]{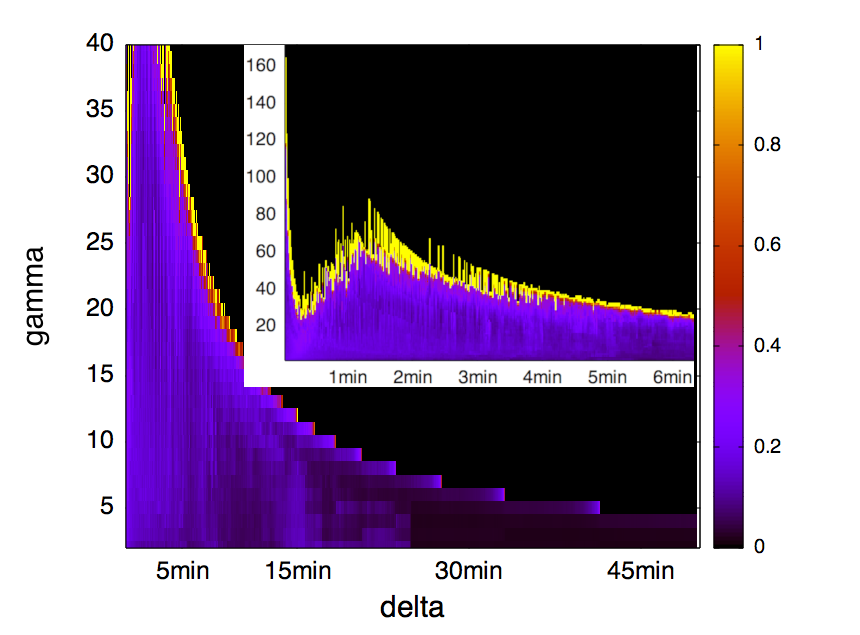}\hfill}

  \caption{\label{fig_approx_quality} Enron dataset (left) and Rollernet dataset (right): ratio obtained by dividing the $2$-approximation output by the number of $\gamma$-edges in the kernelization output. Here, the brighter is the better: the approximation solves \textsc{$\gamma$-matching} optimally on regions where the ratio is $100\%$, which is denoted by the yellow color.
  The parameter $k$ for the kernelization algorithm is the size of the solution found by the $2$-approximation algorithm.}
\end{figure}

\textit{Observations w.r.t.\ Hypothesis 2. Kernelization quality:}\\
Results are given in Fig.~\ref{fig_kern_quality}.
The parameter $k$ for the kernelization algorithm is the size of the solution found by the $2$-approximation algorithm.
We observe that on well chosen intervals of $\delta$ and $\gamma$, kernelization reduces the input size down to under twenty per cent.
This is particularly true for $\gamma\approx20, 30$ with $\delta\leq 2$ month on Enron;
and $\gamma\leq10$ with $\delta\leq20$ minutes for Rollernet.
We observe on these values that the kernelization algorithm reduces the input link stream to an equivalent instance of size smaller than $20\%$ the size of the original input, and sometimes under $10\%$.
We conclude that Hypothesis 2 is sound.

\textit{Observations w.r.t.\ Hypothesis 3. Approximation quality:}\\
Results are given in Fig.~\ref{fig_approx_quality}.
The parameter $k$ for the kernelization algorithm is the size of the solution found by the $2$-approximation algorithm.
We notice from definition that the $2$-approximation algorithm produces a lower bound for \textsc{$\gamma$-matching}, which is at least half the optimal value.
Moreover, the kernelization algorithm gives a naive upper bound for \textsc{$\gamma$-matching} by simply counting the number of $\gamma$-edges present in the kernel.
We observe for tangent areas in Fig.~\ref{fig_approx_quality} that these two upper and lower bounds meet.
This means that the $2$-approximation outputs an optimal solution for \textsc{$\gamma$-matching} on these areas.
However, we observe that for most parts of our dataset, Hypothesis 3 is not confirmed.
At this state, our experiments w.r.t.\ Hypothesis 3 fail in giving any clue for a conclusion.
Further experiments must be done in order to clarify this question.
Our feeling, however, is that Hypothesis 3 is generally false.
We conjecture that the $2$ approximation factor is far from optimal.

\subsection{Discussion}
\label{subsec_discussion}
Our experiment results are optimistic about the numerical usefulness of kernelization in finding temporal matching.
At the same time, they also point out several questions.
While our experiments allow us to observe that:
\begin{itemize}
  \item preprocessing an instance of \textsc{$\gamma$-matching} by a greedy process, and then kernelization as described in Theroem~\ref{maintheorem}, seems to be sound;
  \item the preprocessing is very quick on real world input;
  \item the preprocessing is robust versus stress testing on large inputs using common laptop: below ten seconds on input of hundreds thousand timed edges;
\end{itemize}
they also testify that works still need to be done for further investigating \textsc{$\gamma$-matching}, especially that:
\begin{itemize}
  \item the optimisation problem is NP-hard;
  \item numerical proofs of optimality of the $2$-approximation are only available for very marginal bits of data in the Enron and Rollernet datasets;
  \item while it is true that the kernelization algorithm helps in reducing the input down to $10-20\%$ for interesting mining parameters on Enron and Rollernet datasets, we still do not know how then to find a $\gamma$-matching of the kernel that is better than the output of the $2$-approximation;
  \item in particular, we do not know if the approximation factor can be improved.
\end{itemize}

\section{Conclusion and perspectives}\label{sec:c}
We introduce the notion of temporal matching in a link stream.
Unfortunately, the problem of computing a temporal matching of maximum size, called \textsc{$\gamma$-matching}, turns out to be NP-hard.
We then show a kernelization algorithm for \textsc{$\gamma$-matching} parameterized by the size of the solution.
Our process produces quadratic kernels.
On the way to obtaining the kernelization algorithm, we also provide a $2$-approximation algorithm for \textsc{$\gamma$-matching}.
We believe that the same techniques extend to a large class of hitting set problems in link streams.

\medskip

\noindent{\bf\textit{Acknowledgements:}}
We are grateful to Cl\'emence Magnien for helpful pointers and valuable pieces of advice.
We are grateful to the anonymous referees for helpful comments which greatly improve the paper.
In particular, we are grateful to the referee who coined the name of time-compression, whose usage greatly improves the paper.
  For financial support, we are grateful to:
  \textit{Thales Communications \& Security}, project TCS.DJ.2015-432;
  \textit{Agence Nationale de la Recherche Technique}, project 2016.0097;
  \textit{Centre National de la Recherche Scientifique}, project INS2I.GraphGPU.

\bibliography{main}
\bibliographystyle{plain}

\end{document}